\pdfoutput=1
\documentclass[english]{article}
\usepackage[T1]{fontenc}
\usepackage[latin9]{inputenc}
\pdfoutput=1
\usepackage{geometry}
\usepackage{xcolor}
\usepackage{amsmath}
\usepackage{amsthm}
\usepackage{amssymb}
\usepackage{authblk}
\makeatletter
\theoremstyle{remark}
\newtheorem{rem}{\protect\remarkname}
\theoremstyle{plain}
\newtheorem{lem}{\protect\lemmaname}
\theoremstyle{plain}
\newtheorem{thm}{\protect\theoremname}
\theoremstyle{remark}
\newtheorem{assumption}{Assumption}
\usepackage{appendix}
\providecommand{\remarkname}{Remark}
\providecommand{\theoremname}{Theorem}
\makeatother
\usepackage{babel}
\providecommand{\lemmaname}{Lemma}
\providecommand{\remarkname}{Remark}
\providecommand{\theoremname}{Theorem}
\begin{document}

\title{Bounds on Deep Neural Network Partial Derivatives with Respect to Parameters}

\author[1]{Omkar Sudhir Patil}
\author[1]{Brandon C. Fallin}
\author[1]{Cristian F. Nino}
\author[1]{Rebecca G. Hart}
\author[1]{Warren E. Dixon}
\affil[1]{Department of Mechanical and Aerospace Engineering, University of Florida, Gainesville FL 32611-6250 USA\\
Email: \texttt{\{patilomkarsudhir, brandonfallin, cristian1928, rebecca.hart, wdixon\}@ufl.edu}}

\maketitle

\begin{abstract}
Deep neural networks (DNNs) have emerged as a powerful tool with a
growing body of literature exploring Lyapunov-based approaches for
real-time system identification and control. These methods depend
on establishing bounds for the second partial derivatives of DNNs
with respect to their parameters, a requirement often assumed but
rarely addressed explicitly. This paper provides rigorous mathematical
formulations of polynomial bounds on both the first and second partial
derivatives of DNNs with respect to their parameters. We present lemmas
that characterize these bounds for fully-connected DNNs, while accommodating
various classes of activation function including sigmoidal and ReLU-like
functions. Our analysis yields closed-form expressions that enable
precise stability guarantees for Lyapunov-based deep neural networks
(Lb-DNNs). Furthermore, we extend our results to bound the higher-order
terms in first-order Taylor approximations of DNNs, providing important
tools for convergence analysis in gradient-based learning algorithms.
The resulting framework replaces previously assumed bounds with
explicit, computable expressions, thereby strengthening the mathematical
foundation of neural network applications in safety-critical
control systems.
\end{abstract}

\section{Introduction}

The integration of deep learning techniques into control systems has
emerged as a promising research direction, with numerous recent studies
exploring Lyapunov-based deep learning approaches for real-time system
identification and control \cite{Patil.Le.ea2022,Le.Patil.ea25,Patil.Le.ea.2022,Hart.patil.ea2023,Nino.Patil.ea2023,Vacchini.Sacchi.ea2023,Mei.Jian.2024,Lu.Wu.2024,Chen.Mei.ea2024,Nino.Patil.ea2025a}.
An important component of these approaches involves Lyapunov-based
stability analysis, which is based on the existence of known bounds
for the second partial derivatives of deep neural networks (DNNs)
with respect to their parameters. Despite the theoretical significance
of these bounds, their explicit mathematical formulation has remained
largely unaddressed in the literature.

This paper addresses this gap by deriving rigorous mathematical lemmas
that establish explicit polynomial bounds on both first and second
partial derivatives of fully-connected DNNs with respect to their
parameters. These results provide the theoretical foundation necessary
for guaranteeing stability in Lyapunov-based deep learning control
systems, transforming previously assumed bounds into computable
expressions.

For completeness, we present a detailed description
of the deep neural network architecture commonly employed in Lyapunov-based
deep learning results. While our analysis focuses on this standard
fully-connected feedforward DNN architecture, we note that the methodologies
and results developed herein can be readily extended to accommodate
various neural network architectures.

\section{Preliminaries on Deep Neural Networks}

Given some matrix $A\triangleq\left[a_{i,j}\right]\in\mathbb{R}^{n\times m}$,
where $a_{i,j}$ denotes the element in the $i^{th}$ row and $j^{th}$
column of $A$, the vectorization operator is defined as $\mathrm{vec}(A)\triangleq[a_{1,1},\ldots,a_{n,1},\ldots,a_{1,m},\ldots,a_{n,m}]^{\top}\in\mathbb{R}^{nm}$.
In the following development, we consider a fully-connected deep neural
network (DNN) $\Phi:\mathbb{R}^{L_{\textrm{in}}}\times\mathbb{R}^{p}\to\mathbb{R}^{L_{\textrm{out}}}$
with $k\in\mathbb{Z}_{>0}$ hidden layers, input size $L_{\textrm{in}}\in\mathbb{Z}_{>0}$,
output size $L_{\textrm{out}}\in\mathbb{Z}_{>0}$, and total number
of parameters $p\in\mathbb{Z}_{>0}$, where the parameters
include weights and bias terms. Let $\sigma\in\mathbb{R}^{L_{\textrm{in}}}$
denote the DNN input and $\theta\in\mathbb{R}^{p}$ denote the concatenated
vector of DNN parameters. Then, a fully-connected feedforward DNN
$\Phi(\sigma,\theta)$ is defined using a recursive relation $\Phi_{j}\in\mathbb{R}^{L_{j+1}}$
given by
\begin{eqnarray}
\Phi_{j} & \triangleq & \begin{cases}
V_{j}^{\top}\phi_{j}\left(\Phi_{j-1}\right), & j\in\left\{ 1,\ldots,k\right\} ,\\
V_{j}^{\top}\sigma_{a}, & j=0,
\end{cases}\label{eq:phij_dnn}
\end{eqnarray}
where $\Phi(\sigma,\theta)=\Phi_{k}$, and $\sigma_{a}\triangleq\left[\begin{array}{cc}
\sigma^{\top} & 1\end{array}\right]^{\top}$ denotes the augmented input that accounts for the bias terms, $V_{j}\in\mathbb{R}^{L_{j}\times L_{j+1}}$
denotes the matrix of weights and biases, and $L_{j}\in\mathbb{Z}_{>0}$
denotes the number of neurons in the $j^{\textrm{th}}$ layer for
all $j\in\left\{ 0,\ldots,k\right\} $ with $L_{0}\triangleq L_{\textrm{in}}+1$
and $L_{k+1}\triangleq L_{\textrm{out}}$. For notational convenience,
we write $n\triangleq L_{\textrm{out}}$ throughout. The vector of smooth activation
functions is denoted by $\phi_{j}:\mathbb{R}^{L_{j}}\to\mathbb{R}^{L_{j}}$
for all $j\in\left\{ 1,\ldots,k\right\} $. The activation functions
at each layer are represented as $\phi_{j}\triangleq\left[\begin{array}{cccc}
\varsigma_{j,1} & \ldots & \varsigma_{j,L_{j}-1} & 1\end{array}\right]^{\top}$, where $\varsigma_{j,i}:\mathbb{R}\to\mathbb{R}$ denotes the activation
function at the $i^{\mathrm{th}}$ node of the $j^{\mathrm{th}}$
layer. The trailing constant entry of $1$ in $\phi_{j}$ is a bias-augmentation
convention that propagates the bias term through subsequent layers. For the DNN architecture in (\ref{eq:phij_dnn}), the vector
of DNN weights is defined as $\theta\triangleq\left[\begin{array}{ccc}
\mathrm{vec}(V_{0})^{\top} & \ldots & \mathrm{vec}(V_{k})^{\top}\end{array}\right]^{\top}$ with size $p=\Sigma_{j=0}^{k}L_{j}L_{j+1}$. The Jacobian of the
activation function vector at the $j^{\mathrm{th}}$ layer is denoted
by $\phi_{j}^{\prime}:\mathbb{R}^{L_{j}}\to\mathbb{R}^{L_{j}\times L_{j}}$
and is defined as $\phi_{j}^{\prime}(y)\triangleq\frac{\partial}{\partial y}\phi_{j}(y)$.
Specifically, $\phi_{j}^{\prime}$ evaluates as $\phi_{j}^{\prime}=\mathrm{diag}\left(\left[\begin{array}{cccc}
\varsigma_{j,1}^{\prime} & \ldots & \varsigma_{j,L_{j}-1}^{\prime} & 0\end{array}\right]^{\top}\right)$, where $\varsigma_{j,i}^{\prime}(\zeta)\triangleq\frac{\partial}{\partial\zeta}\varsigma_{j,i}(\zeta)$,
and $\mathrm{diag}(\cdot)$ represents the diagonalization operation
which returns a diagonal matrix with the elements of its input vector
arranged along the diagonal. The Jacobian of the DNN with respect
to the weights is denoted by $\frac{\partial}{\partial\theta}\Phi(\sigma,\theta)$
is represented as
\[
\frac{\partial}{\partial\theta}\Phi(\sigma,\theta)=\left[\begin{array}{cccc}
\frac{\partial\Phi\left(\sigma,\theta\right)}{\partial\mathrm{vec}(V_{0})}, & \frac{\partial\Phi\left(\sigma,\theta\right)}{\partial\mathrm{vec}(V_{1})}, & \ldots, & \frac{\partial\Phi\left(\sigma,\theta\right)}{\partial\mathrm{vec}(V_{k})}\end{array}\right]\in\mathbb{R}^{n\times p},
\]
where $\frac{\partial\Phi\left(\sigma,\theta\right)}{\partial\mathrm{vec}(V_{j})}\in\mathbb{R}^{n\times L_{j}L_{j+1}}$.
Then, applying the property
$\frac{\partial}{\partial\mathrm{vec}(B)}\mathrm{vec}(ABC)=C^{\top}\otimes A$ to the DNN architecture in (\ref{eq:phij_dnn})
yields
\begin{equation}
\frac{\partial\Phi\left(\sigma,\theta\right)}{\partial\mathrm{vec}(V_{j})}=\left(\overset{\curvearrowleft}{\prod_{l=j+1}^{k}}V_{l}^{\top}\phi_{l}^{\prime}\left(\Phi_{l-1}\right)\right)\left(I_{L_{j+1}}\otimes\varphi_{j}^{\top}\right),\label{eq:Phiprime_j}
\end{equation}
where $\varphi_{j}$ is a shorthand notation defined as $\varphi_{0}\triangleq\sigma_{a}$
and $\varphi_{j}\triangleq\phi_{j}\left(\Phi_{j-1}\right)$ for all
$j\in\{1,\ldots,k\}$. In (\ref{eq:Phiprime_j}), the notation $\stackrel{\curvearrowleft}{\prod}$
denotes the right-to-left matrix product operation, i.e., ${\displaystyle \overset{\curvearrowleft}{\prod}_{p=1}^{m}A_{p}=A_{m}\ldots A_{2}A_{1}}$
and ${\displaystyle \overset{\curvearrowleft}{\prod}_{p=a}^{m}A_{p}=I}$
if $a>m$, and $\otimes$ denotes the Kronecker product. To facilitate
the subsequent development and analysis, the following assumption
is made regarding the activation functions of the DNN.

\begin{assumption}
\label{assm:activation bounds} For each $j\in\left\{ 1,\ldots,k\right\} $,
the activation function $\phi_{j}$, its Jacobian $\phi_{j}^{\prime}$,
and Hessian $\phi_{j}^{\prime\prime}\left(y\right)\triangleq\frac{\partial^{2}}{\partial y^{2}}\phi_{j}\left(y\right)$
are bounded as
\begin{eqnarray}
\left\Vert \phi_{j}\left(y\right)\right\Vert  & \leq & \mathfrak{a}_{1}\left\Vert y\right\Vert +\mathfrak{a}_{0},\nonumber \\
\left\Vert \phi_{j}^{\prime}\left(y\right)\right\Vert  & \leq & \mathfrak{b}_{0},\nonumber \\
\left\Vert \phi_{j}^{\prime\prime}\left(y\right)\right\Vert  & \leq & \mathfrak{c}_{0},\label{eq:Activation Bounds}
\end{eqnarray}
respectively, where $\mathfrak{a}_{0},\mathfrak{a}_{1},\mathfrak{b}_{0},\mathfrak{c}_{0}\in\mathbb{R}_{>0}$
are known constants.
\end{assumption}

\begin{rem}
\label{rem:activation bounds} Most activation functions used in practice
satisfy Assumption \ref{assm:activation bounds}. Specifically, sigmoidal
activation functions (e.g., logistic function, hyperbolic tangent,
etc.) have $\left\Vert \phi_{j}\left(y\right)\right\Vert $, $\left\Vert \phi_{j}^{\prime}\left(y\right)\right\Vert $,
and $\left\Vert \phi_{j}^{\prime\prime}\left(y\right)\right\Vert $
bounded uniformly by constants. Smooth approximations of rectified
linear units (ReLUs) such as Swish grow linearly, and hence satisfy
the bound $\left\Vert \phi_{j}\left(y\right)\right\Vert \leq\mathfrak{a}_{1}\left\Vert y\right\Vert +\mathfrak{a}_{0}$
of Assumption \ref{assm:activation bounds}. Exact ReLU activations
are excluded because they are not twice differentiable, but the results
apply to any smooth approximation satisfying the bounds in (\ref{eq:Activation Bounds}).
\end{rem}

\section{Bounding Analysis for DNN Layers and Partial Derivatives}

The objective is to obtain polynomial bounds on the DNN layers and
their first and second partial derivatives with respect to $\theta$.
In this section, Lemma \ref{lem:DNN bound} provides a bound on any
arbitrary DNN layer, and Lemmas \ref{lem:Jacobian bound} and \ref{lem:Hessian bounds}
provide bounds on the first and second partial derivatives of the DNN layers with respect
to $\theta$, respectively.
\begin{lem}
\label{lem:DNN bound} For the DNN architecture described in (\ref{eq:phij_dnn}),
the output of the $j^{th}$ layer of the DNN, $\Phi_{j}$, is bounded
as
\begin{equation}
\left\Vert \Phi_{j}\right\Vert \leq\mathfrak{a}_{1}^{j}\left\Vert \sigma_{a}\right\Vert \prod_{i=0}^{j}\left\Vert V_{i}\right\Vert +\mathfrak{a}_{0}\sum_{i=1}^{j}\left(\mathfrak{a}_{1}^{j-i}\prod_{l=i}^{j}\left\Vert V_{l}\right\Vert \right),\label{eq:Phi_j bound}
\end{equation}
for all $j\in\left\{ 0,\ldots,k\right\} $, and the corresponding
activation $\phi_{j}\left(\Phi_{j-1}\right)$ is bounded as
\begin{equation}
\left\Vert \phi_{j}\left(\Phi_{j-1}\right)\right\Vert \leq\mathfrak{a}_{1}^{j}\left\Vert \sigma_{a}\right\Vert \prod_{i=0}^{j-1}\left\Vert V_{i}\right\Vert +\mathfrak{a}_{0}\sum_{i=1}^{j-1}\left(\mathfrak{a}_{1}^{j-i}\prod_{l=i}^{j-1}\left\Vert V_{l}\right\Vert \right)+\mathfrak{a}_{0},\label{eq:phi_j bound}
\end{equation}
for all $j\in\left\{ 1,\ldots,k\right\} $. Furthermore, if the bound
$\left\Vert V_{j}\right\Vert \leq\bar{\theta}$ is applied for all
$j\in\left\{ 0,\ldots,k\right\} $, it follows that $\left\Vert \Phi_{j}\right\Vert \leq\mathfrak{a}_{1}^{j}\bar{\theta}^{j+1}\left\Vert \sigma_{a}\right\Vert +\mathfrak{a}_{0}\sum_{i=1}^{j}\mathfrak{a}_{1}^{j-i}\bar{\theta}^{j-i+1}$.
In addition, if the activation functions are uniformly bounded by
constants, i.e., $\mathfrak{a}_{1}=0$ in Assumption \ref{assm:activation bounds},
then the bound is further simplified to $\left\Vert \Phi_{0}\right\Vert \leq\bar{\theta}\left\Vert \sigma_{a}\right\Vert $
and $\left\Vert \Phi_{j}\right\Vert \leq\mathfrak{a}_{0}\bar{\theta}$
for all $j\in\left\{ 1,\ldots,k\right\} $.
\end{lem}
\begin{proof}
Consider the base case of (\ref{eq:Phi_j bound}) when $j=0$ for
mathematical induction. Substituting $j=0$ into (\ref{eq:phij_dnn})
yields $\Phi_{0}=V_{0}^{\top}\sigma_{a}$, which is bounded as
\[
\left\Vert \Phi_{0}\right\Vert \leq\left\Vert V_{0}\right\Vert \left\Vert \sigma_{a}\right\Vert .
\]
Hence, (\ref{eq:Phi_j bound}) holds for the base case.

Assume for induction that the bound in (\ref{eq:Phi_j bound}) holds
for $\left\Vert \Phi_{j-1}\right\Vert $ for all $j\in\left\{ 1,\ldots,k\right\} $,
i.e.,
\begin{equation}
\left\Vert \Phi_{j-1}\right\Vert \leq\mathfrak{a}_{1}^{j-1}\left\Vert \sigma_{a}\right\Vert \prod_{i=0}^{j-1}\left\Vert V_{i}\right\Vert +\mathfrak{a}_{0}\sum_{i=1}^{j-1}\left(\mathfrak{a}_{1}^{j-i-1}\prod_{l=i}^{j-1}\left\Vert V_{l}\right\Vert \right).\label{eq:Phi_j-1 bound}
\end{equation}
Recall from (\ref{eq:phij_dnn}) that $\Phi_{j}=V_{j}^{\top}\phi_{j}\left(\Phi_{j-1}\right)$
for all $j\in\left\{ 1,\ldots,k\right\} $. Therefore, it follows
that
\begin{equation}
\left\Vert \Phi_{j}\right\Vert \leq\left\Vert V_{j}\right\Vert \left\Vert \phi_{j}\left(\Phi_{j-1}\right)\right\Vert .\label{eq:Phi j induction bound}
\end{equation}
Applying (\ref{eq:Activation Bounds}) gives $\left\Vert \phi_{j}\left(\Phi_{j-1}\right)\right\Vert \leq\mathfrak{a}_{1}\left\Vert \Phi_{j-1}\right\Vert +\mathfrak{a}_{0}$.
Substituting the inductive assumption (\ref{eq:Phi_j-1 bound}) into this expression yields
\begin{eqnarray}
\left\Vert \phi_{j}\left(\Phi_{j-1}\right)\right\Vert  & \leq & \mathfrak{a}_{1}\left[\mathfrak{a}_{1}^{j-1}\left\Vert \sigma_{a}\right\Vert \prod_{i=0}^{j-1}\left\Vert V_{i}\right\Vert +\mathfrak{a}_{0}\sum_{i=1}^{j-1}\left(\mathfrak{a}_{1}^{j-i-1}\prod_{l=i}^{j-1}\left\Vert V_{l}\right\Vert \right)\right]+\mathfrak{a}_{0}\nonumber \\
 & = & \mathfrak{a}_{1}^{j}\left\Vert \sigma_{a}\right\Vert \prod_{i=0}^{j-1}\left\Vert V_{i}\right\Vert +\mathfrak{a}_{0}\sum_{i=1}^{j-1}\left(\mathfrak{a}_{1}^{j-i}\prod_{l=i}^{j-1}\left\Vert V_{l}\right\Vert \right)+\mathfrak{a}_{0},\label{eq:phi_j bound expansion}
\end{eqnarray}
which establishes (\ref{eq:phi_j bound}). Substituting (\ref{eq:phi_j bound expansion})
into (\ref{eq:Phi j induction bound}) and distributing $\left\Vert V_{j}\right\Vert $
gives
\begin{equation}
\left\Vert \Phi_{j}\right\Vert \leq\mathfrak{a}_{1}^{j}\left\Vert \sigma_{a}\right\Vert \prod_{i=0}^{j}\left\Vert V_{i}\right\Vert +\mathfrak{a}_{0}\sum_{i=1}^{j-1}\left(\mathfrak{a}_{1}^{j-i}\prod_{l=i}^{j}\left\Vert V_{l}\right\Vert \right)+\mathfrak{a}_{0}\left\Vert V_{j}\right\Vert .\label{eq:Phi_j expansion}
\end{equation}
The trailing term $\mathfrak{a}_{0}\left\Vert V_{j}\right\Vert $ corresponds to
the $i=j$ summand of the sum in (\ref{eq:Phi_j bound}), since at $i=j$
the summand evaluates to $\mathfrak{a}_{1}^{j-j}\prod_{l=j}^{j}\left\Vert V_{l}\right\Vert =\left\Vert V_{j}\right\Vert $.
Absorbing this term into the sum extends the upper limit from $j-1$ to $j$,
yielding (\ref{eq:Phi_j bound}). Hence, by mathematical induction,
(\ref{eq:Phi_j bound}) and (\ref{eq:phi_j bound}) hold for all $j\in\left\{ 1,\ldots,k\right\} $.

Furthermore, if the bound $\left\Vert V_{j}\right\Vert \leq\bar{\theta}$
is applied for all $j\in\left\{ 0,\ldots,k\right\} $, the products
in (\ref{eq:Phi_j bound}) are bounded by counting factors:
$\prod_{i=0}^{j}\left\Vert V_{i}\right\Vert $ contains $j+1$ factors and is therefore
bounded by $\bar{\theta}^{j+1}$, while $\prod_{l=i}^{j}\left\Vert V_{l}\right\Vert $
contains $j-i+1$ factors and is bounded by $\bar{\theta}^{j-i+1}$.
Substituting these bounds into (\ref{eq:Phi_j bound}) yields
$\left\Vert \Phi_{j}\right\Vert \leq\mathfrak{a}_{1}^{j}\bar{\theta}^{j+1}\left\Vert \sigma_{a}\right\Vert +\mathfrak{a}_{0}\sum_{i=1}^{j}\mathfrak{a}_{1}^{j-i}\bar{\theta}^{j-i+1}$.
In addition, if the activation functions are uniformly bounded by
constants, i.e., $\mathfrak{a}_{1}=0$ in Assumption \ref{assm:activation bounds},
then it follows that $\left\Vert \Phi_{j}\right\Vert \leq\left\Vert V_{j}\right\Vert \left\Vert \phi_{j}\left(\Phi_{j-1}\right)\right\Vert \leq\left\Vert V_{j}\right\Vert \mathfrak{a}_{0}\leq\bar{\theta}\mathfrak{a}_{0}$
for all $j\in\{1,\ldots,k\}$.
\end{proof}

\begin{lem}
\label{lem:Jacobian bound} For the DNN architecture described in
(\ref{eq:phij_dnn}), the Jacobian of the $w^{th}$ layer with respect
to the $j^{th}$ layer weights is bounded as
% CORRECTION: case label "w <= j" corrected to "w >= j".
\begin{equation}
\left\Vert \frac{\partial\Phi_{w}}{\partial\mathrm{vec}\left(V_{j}\right)}\right\Vert \leq\begin{cases}
\begin{aligned}
 & \mathfrak{b}_{0}^{w-j}\left(\prod_{l=j+1}^{w}\left\Vert V_{l}\right\Vert \right)\\
 & \quad\bigg(\mathfrak{a}_{1}^{j}\left\Vert \sigma_{a}\right\Vert \prod_{i=0}^{j-1}\left\Vert V_{i}\right\Vert \\
 & \quad\quad+\mathfrak{a}_{0}\sum_{i=1}^{j-1}\Big(\mathfrak{a}_{1}^{j-i}\prod_{l=i}^{j-1}\left\Vert V_{l}\right\Vert \Big)+\mathfrak{a}_{0}\bigg),
\end{aligned} & w\geq j,\\
0, & j>w,
\end{cases}\label{eq:Jacobian w bound}
\end{equation}
for all $w,j\in\left\{ 0,\ldots,k\right\} $. Furthermore, if the
bound $\left\Vert V_{j}\right\Vert \leq\bar{\theta}$ holds for all
$j\in\left\{ 0,\ldots,k\right\} $, then the Jacobian $\frac{\partial\Phi\left(\sigma,\theta\right)}{\partial\theta}$
is bounded as
% CORRECTION: exponent on first bar{theta} was j+1, should be j;
% exponent inside the inner sum was j-i-1, should be j-i.
\begin{equation}
\begin{aligned}
\left\Vert \frac{\partial\Phi\left(\sigma,\theta\right)}{\partial\theta}\right\Vert  & \leq\mathfrak{b}_{0}^{k}\bar{\theta}^{k}\left\Vert \sigma_{a}\right\Vert +\sum_{j=1}^{k}\mathfrak{b}_{0}^{k-j}\bar{\theta}^{k-j}\bigg(\mathfrak{a}_{1}^{j}\left\Vert \sigma_{a}\right\Vert \bar{\theta}^{j}\\
 & \qquad+\mathfrak{a}_{0}\sum_{i=1}^{j-1}\Big(\mathfrak{a}_{1}^{j-i}\bar{\theta}^{j-i}\Big)+\mathfrak{a}_{0}\bigg).
\end{aligned}
\label{eq:Jacobian bound 1}
\end{equation}
In addition, if the activations are uniformly bounded by constants,
i.e., $\mathfrak{a}_{1}=0$ in Assumption \ref{assm:activation bounds},
then the Jacobian bound reduces to
\begin{equation}
\left\Vert \frac{\partial\Phi\left(\sigma,\theta\right)}{\partial\theta}\right\Vert \leq\mathfrak{b}_{0}^{k}\bar{\theta}^{k}\left\Vert \sigma_{a}\right\Vert +\mathfrak{a}_{0}\sum_{j=1}^{k}\mathfrak{b}_{0}^{k-j}\bar{\theta}^{k-j}.\label{eq:Jacobian bound 2}
\end{equation}
\end{lem}
\begin{proof}
Replacing $k$ with $w$ in (\ref{eq:Phiprime_j}) yields
\begin{equation}
\frac{\partial\Phi_{w}}{\partial\mathrm{vec}\left(V_{j}\right)}=\begin{cases}
\left(\overset{\curvearrowleft}{\prod_{l=j+1}^{w}}V_{l}^{\top}\phi_{l}^{\prime}\left(\Phi_{l-1}\right)\right)\left(I_{L_{j+1}}\otimes\varphi_{j}^{\top}\right), & w\geq j,\\
0, & j>w,
\end{cases}\label{eq:Jacobian bound j}
\end{equation}
where $\frac{\partial\Phi_{w}}{\partial\mathrm{vec}\left(V_{j}\right)}=0$
if $j>w$ because the outputs of the inner layers do not depend on
the outer layer weights. For $w\geq j$, taking norms of (\ref{eq:Jacobian bound j})
and applying submultiplicativity of the spectral norm yields
\begin{equation}
\left\Vert \frac{\partial\Phi_{w}}{\partial\mathrm{vec}\left(V_{j}\right)}\right\Vert \leq\left\Vert \overset{\curvearrowleft}{\prod_{l=j+1}^{w}}V_{l}^{\top}\phi_{l}^{\prime}\left(\Phi_{l-1}\right)\right\Vert \left\Vert I_{L_{j+1}}\otimes\varphi_{j}^{\top}\right\Vert .\label{eq:Jacobian split}
\end{equation}
The Kronecker product satisfies $\left\Vert I_{n}\otimes A\right\Vert =\left\Vert A\right\Vert $
in the spectral norm, since $I_{n}\otimes A$ is block-diagonal with $n$ copies of $A$.
Therefore $\left\Vert I_{L_{j+1}}\otimes\varphi_{j}^{\top}\right\Vert =\left\Vert \varphi_{j}\right\Vert $.
Applying submultiplicativity to the matrix product, together with the bound
$\left\Vert \phi_{l}^{\prime}\left(\Phi_{l-1}\right)\right\Vert \leq\mathfrak{b}_{0}$
from Assumption \ref{assm:activation bounds} and the identity $\left\Vert V_{l}^{\top}\right\Vert =\left\Vert V_{l}\right\Vert $, gives
\begin{equation}
\left\Vert \overset{\curvearrowleft}{\prod_{l=j+1}^{w}}V_{l}^{\top}\phi_{l}^{\prime}\left(\Phi_{l-1}\right)\right\Vert \leq\prod_{l=j+1}^{w}\left\Vert V_{l}\right\Vert \left\Vert \phi_{l}^{\prime}\left(\Phi_{l-1}\right)\right\Vert \leq\mathfrak{b}_{0}^{w-j}\prod_{l=j+1}^{w}\left\Vert V_{l}\right\Vert .\label{eq:product chain bound}
\end{equation}
The factor $\left\Vert \varphi_{j}\right\Vert $ is bounded by Lemma \ref{lem:DNN bound}: for $j\geq1$,
$\varphi_{j}=\phi_{j}\left(\Phi_{j-1}\right)$, so (\ref{eq:phi_j bound}) applies directly,
while for $j=0$ the identity $\varphi_{0}=\sigma_{a}$ gives $\left\Vert \varphi_{0}\right\Vert =\left\Vert \sigma_{a}\right\Vert $.
Combining (\ref{eq:Jacobian split}), (\ref{eq:product chain bound}),
and the bound on $\left\Vert \varphi_{j}\right\Vert $ yields
\[
\left\Vert \frac{\partial\Phi_{w}}{\partial\mathrm{vec}\left(V_{j}\right)}\right\Vert \leq\begin{cases}
\begin{array}{c}
\mathfrak{b}_{0}^{w-j}\left(\prod_{l=j+1}^{w}\left\Vert V_{l}\right\Vert \right)\left(\mathfrak{a}_{1}^{j}\left\Vert \sigma_{a}\right\Vert \prod_{i=0}^{j-1}\left\Vert V_{i}\right\Vert \right.\\
\left.+\mathfrak{a}_{0}\sum_{i=1}^{j-1}\left(\mathfrak{a}_{1}^{j-i}\prod_{l=i}^{j-1}\left\Vert V_{l}\right\Vert \right)+\mathfrak{a}_{0}\right),
\end{array} & w\geq j,\\
0, & j>w.
\end{cases}
\]
Furthermore, using the triangle inequality, the Jacobian is bounded as
\[
\left\Vert \frac{\partial\Phi\left(\sigma,\theta\right)}{\partial\theta}\right\Vert \leq\left\Vert \frac{\partial\Phi_{k}}{\partial\mathrm{vec}\left(V_{0}\right)}\right\Vert +\left\Vert \frac{\partial\Phi_{k}}{\partial\mathrm{vec}\left(V_{1}\right)}\right\Vert +\ldots+\left\Vert \frac{\partial\Phi_{k}}{\partial\mathrm{vec}\left(V_{k}\right)}\right\Vert .
\]
Therefore,
% CORRECTION: exponent b_0^{j-k} in the source was wrong; should be b_0^{k-j}.
\begin{eqnarray*}
\left\Vert \frac{\partial\Phi\left(\sigma,\theta\right)}{\partial\theta}\right\Vert  & \leq & \mathfrak{b}_{0}^{k}\left(\prod_{l=1}^{k}\left\Vert V_{l}\right\Vert \right)\left\Vert \sigma_{a}\right\Vert \\
 &  & +\sum_{j=1}^{k}\mathfrak{b}_{0}^{k-j}\left(\prod_{l=j+1}^{k}\left\Vert V_{l}\right\Vert \right)\bigg(\mathfrak{a}_{1}^{j}\left\Vert \sigma_{a}\right\Vert \prod_{i=0}^{j-1}\left\Vert V_{i}\right\Vert \\
 &  & \qquad+\mathfrak{a}_{0}\sum_{i=1}^{j-1}\Big(\mathfrak{a}_{1}^{j-i}\prod_{l=i}^{j-1}\left\Vert V_{l}\right\Vert \Big)+\mathfrak{a}_{0}\bigg).
\end{eqnarray*}
As a result, applying the bound $\left\Vert V_{j}\right\Vert \leq\bar{\theta}$
for all $j\in\left\{ 0,\ldots,k\right\} $, the products are bounded by counting factors:
$\prod_{l=1}^{k}\left\Vert V_{l}\right\Vert \leq\bar{\theta}^{k}$
($k$ factors), $\prod_{l=j+1}^{k}\left\Vert V_{l}\right\Vert \leq\bar{\theta}^{k-j}$
($k-j$ factors), $\prod_{i=0}^{j-1}\left\Vert V_{i}\right\Vert \leq\bar{\theta}^{j}$
($j$ factors), and $\prod_{l=i}^{j-1}\left\Vert V_{l}\right\Vert \leq\bar{\theta}^{j-i}$
($j-i$ factors). Substituting these bounds yields
\begin{eqnarray*}
\left\Vert \frac{\partial\Phi\left(\sigma,\theta\right)}{\partial\theta}\right\Vert  & \leq & \mathfrak{b}_{0}^{k}\bar{\theta}^{k}\left\Vert \sigma_{a}\right\Vert +\sum_{j=1}^{k}\mathfrak{b}_{0}^{k-j}\bar{\theta}^{k-j}\bigg(\mathfrak{a}_{1}^{j}\left\Vert \sigma_{a}\right\Vert \bar{\theta}^{j}\\
 &  & \qquad+\mathfrak{a}_{0}\sum_{i=1}^{j-1}\Big(\mathfrak{a}_{1}^{j-i}\bar{\theta}^{j-i}\Big)+\mathfrak{a}_{0}\bigg).
\end{eqnarray*}
Furthermore, if the activation functions are uniformly bounded by
constants, i.e., $\mathfrak{a}_{1}=0$ in Assumption \ref{assm:activation bounds},
then the bound is further simplified as
\[
\left\Vert \frac{\partial\Phi\left(\sigma,\theta\right)}{\partial\theta}\right\Vert \leq\mathfrak{b}_{0}^{k}\bar{\theta}^{k}\left\Vert \sigma_{a}\right\Vert +\mathfrak{a}_{0}\sum_{j=1}^{k}\mathfrak{b}_{0}^{k-j}\bar{\theta}^{k-j}.
\]
\end{proof}

\begin{lem}
\label{lem:Hessian bounds}For the DNN architecture described in (\ref{eq:phij_dnn}),
the second partial derivative term $\frac{\partial^{2}\Phi_{w}}{\partial\mathrm{vec}\left(V_{q}\right)\partial\mathrm{vec}\left(V_{j}\right)}$
is bounded according to
\begin{eqnarray*}
\text{\ensuremath{\left\Vert \frac{\partial^{2}\Phi_{w}}{\partial\mathrm{vec}\left(V_{q}\right)\partial\mathrm{vec}\left(V_{j}\right)}\right\Vert }} & \leq & \mathcal{R}_{w,q,j}\mathcal{Q}_{j}\mathcal{Q}_{q}+\mathcal{T}_{w,j}\mathcal{Q}_{j},
\end{eqnarray*}
where $\mathcal{Q}_{j}$ is defined as
\begin{equation*}
\begin{aligned}
\mathcal{Q}_{j} & \triangleq\mathfrak{a}_{1}^{j}\left\Vert \sigma_{a}\right\Vert \prod_{i=0}^{j-1}\left\Vert V_{i}\right\Vert \\
 & \quad+\mathfrak{a}_{0}\sum_{i=1}^{j-1}\left(\mathfrak{a}_{1}^{j-i}\prod_{l=i}^{j-1}\left\Vert V_{l}\right\Vert \right)+\mathfrak{a}_{0},
\end{aligned}
\end{equation*}
$\mathcal{Q}_{q}$ is obtained by replacing $j$ with $q$ in the
expression for $\mathcal{Q}_{j}$,
\[
\mathcal{T}_{w,j}\triangleq\mathfrak{b}_{0}^{w-j+1}\left(\prod_{l=j+1}^{w}\left\Vert V_{l}\right\Vert \right),
\]
and
\[
\mathcal{R}_{w,q,j}\triangleq\begin{cases}
\begin{aligned}
 & \bigg(\sum_{h=1}^{w-q}\mathfrak{c}_{0}\mathfrak{b}_{0}^{2w-q-j-h-1}\bigg(\prod_{l=q+1}^{w-h}\left\Vert V_{l}\right\Vert \bigg)\bigg)\\
 & \quad\bigg(\prod_{l=j+1}^{w}\left\Vert V_{l}\right\Vert \bigg),
\end{aligned} & j\leq q\leq w,\\
\mathcal{R}_{w,j,q}, & q\leq j\leq w,
\end{cases}
\]
 for all $j,q\leq w\leq k$, implying that $\frac{\partial^{2}\Phi_{w}}{\partial\mathrm{vec}\left(V_{q}\right)\partial\mathrm{vec}\left(V_{j}\right)}$
is bounded by a quadratic polynomial in terms of $\left\Vert \sigma_{a}\right\Vert $.
Furthermore, consider the bound $\left\Vert V_{j}\right\Vert \leq\bar{\theta}$
for all $j\in\left\{ 0,\ldots,k\right\} $. Then, the terms $\mathcal{Q}_{j}$,
$\mathcal{T}_{w,j}$, and $\mathcal{R}_{w,j,q}$ reduce to
% CORRECTION: missing a_0 coefficient in front of the middle sum restored.
\begin{align*}
\mathcal{Q}_{j} & =\mathfrak{a}_{1}^{j}\bar{\theta}^{j}\left\Vert \sigma_{a}\right\Vert +\mathfrak{a}_{0}\sum_{i=1}^{j-1}\left(\mathfrak{a}_{1}^{j-i}\bar{\theta}^{j-i}\right)+\mathfrak{a}_{0},\\
\mathcal{T}_{w,j} & =\mathfrak{b}_{0}^{w-j+1}\bar{\theta}^{w-j},\\
\mathcal{R}_{w,q,j} & =\sum_{h=1}^{w-q}\mathfrak{c}_{0}\mathfrak{b}_{0}^{2w-q-j-h-1}\bar{\theta}^{2w-h-q-j}.
\end{align*}
Furthermore, if the activation functions are uniformly bounded by
constants, i.e., $\mathfrak{a}_{1}=0$ in Assumption \ref{assm:activation bounds},
% CORRECTION: the source stated Q_j = sum(a_1^{j-i} bar{theta}^{j-i}) + a_0 for j >= 1,
% but with a_1 = 0 the sum vanishes, leaving Q_j = a_0 for j >= 1.
then $\mathcal{Q}_{j}$ reduces to $\mathcal{Q}_{0}=\left\Vert \sigma_{a}\right\Vert $
and $\mathcal{Q}_{j}=\mathfrak{a}_{0}$ for all $j\geq1$.
\end{lem}
\begin{proof}
For the ease of subsequent exposition, we consider scalar outputs
for $\Phi_{w}$. This reduction does not compromise generality because
the bounds apply element-wise for the vector case. The second partial
derivative term $\frac{\partial^{2}\Phi_{w}}{\partial\mathrm{vec}\left(V_{q}\right)\partial\mathrm{vec}\left(V_{j}\right)}$
is expressed as
% CORRECTION: added the squared partial notation that was missing in the source.
\begin{eqnarray}
\frac{\partial^{2}\Phi_{w}}{\partial\mathrm{vec}\left(V_{q}\right)\partial\mathrm{vec}\left(V_{j}\right)} & = & \frac{\partial}{\partial\mathrm{vec}\left(V_{q}\right)}\left(\frac{\partial}{\partial\mathrm{vec}\left(V_{j}\right)}V_{w}^{\top}\varphi_{w}\right)\nonumber \\
 & = & \begin{cases}
\begin{aligned}
 & \frac{\partial}{\partial\mathrm{vec}\left(V_{q}\right)}\left(\overset{\curvearrowleft}{\prod_{l=j+1}^{w}}V_{l}^{\top}\phi_{l}^{\prime}\left(\Phi_{l-1}\right)\right)\\
 & \quad\left(I_{L_{j+1}}\otimes\varphi_{j}^{\top}\right),
\end{aligned} & w\geq j,\\
0, & j>w.
\end{cases}\label{eq:partial Hessian w}
\end{eqnarray}
Due to the symmetry of Hessian matrices, it follows that $\frac{\partial^{2}\Phi_{w}}{\partial\mathrm{vec}\left(V_{q}\right)\partial\mathrm{vec}\left(V_{j}\right)}=\frac{\partial^{2}\Phi_{w}}{\partial\mathrm{vec}\left(V_{j}\right)\partial\mathrm{vec}\left(V_{q}\right)}$.
Without loss of generality, consider the case $w\geq q\geq j$. Such
consideration does not affect generality because, if $q<j$, then
the analysis can be performed by exchanging $q$ and $j$. Differentiating
the product chain $\overset{\curvearrowleft}{\prod_{l=j+1}^{w}}V_{l}^{\top}\phi_{l}^{\prime}\left(\Phi_{l-1}\right)$
with respect to $\mathrm{vec}\left(V_{q}\right)$ produces $w-q+1$
nonzero terms by the product rule. For each $l\in\left\{ q+1,\ldots,w\right\} $,
differentiating $\phi_{l}^{\prime}\left(\Phi_{l-1}\right)$ via the
chain rule contributes one term containing $\phi_{l}^{\prime\prime}\left(\Phi_{l-1}\right)\frac{\partial\Phi_{l-1}}{\partial\mathrm{vec}\left(V_{q}\right)}$,
and one additional term arises from differentiating the factor $V_{q}^{\top}$
directly. The factors $\phi_{l}^{\prime}\left(\Phi_{l-1}\right)$
for $l\in\left\{ j+1,\ldots,q\right\} $ contribute zero because $\Phi_{l-1}$
does not depend on $V_{q}$ when $l-1<q$. Applying
the product rule and chain rule to (\ref{eq:partial Hessian w}) yields
\begin{eqnarray*}
\frac{\partial^{2}\Phi_{w}}{\partial\mathrm{vec}\left(V_{q}\right)\partial\mathrm{vec}\left(V_{j}\right)} & = & V_{w}^{\top}\left(\phi_{w}^{\prime\prime}\left(\Phi_{w-1}\right)\frac{\partial\Phi_{w-1}}{\partial\mathrm{vec}\left(V_{q}\right)}\right)\\
 &  & V_{w-1}^{\top}\phi_{w-1}^{\prime}\left(\Phi_{w-2}\right)\ldots\\
 &  & V_{j+1}^{\top}\phi_{j+1}^{\prime}\left(\Phi_{j}\right)\left(I_{L_{j+1}}\otimes\varphi_{j}^{\top}\right)\\
 &  & +V_{w}^{\top}\phi_{w}^{\prime}\left(\Phi_{w-1}\right)V_{w-1}^{\top}\left(\phi_{w-1}^{\prime\prime}\left(\Phi_{w-2}\right)\frac{\partial\Phi_{w-2}}{\partial\mathrm{vec}\left(V_{q}\right)}\right)\\
 &  & V_{w-2}^{\top}\phi_{w-2}^{\prime}\left(\Phi_{w-3}\right)\ldots\\
 &  & V_{j+1}^{\top}\phi_{j+1}^{\prime}\left(\Phi_{j}\right)\left(I_{L_{j+1}}\otimes\varphi_{j}^{\top}\right)\\
 &  & \vdots\\
 &  & +V_{w}^{\top}\phi_{w}^{\prime}\left(\Phi_{w-1}\right)\ldots V_{q+1}^{\top}\left(\phi_{q+1}^{\prime\prime}\left(\Phi_{q}\right)\frac{\partial\Phi_{q}}{\partial\mathrm{vec}\left(V_{q}\right)}\right)\\
 &  & V_{q}^{\top}\phi_{q}^{\prime}\left(\Phi_{q-1}\right)\ldots\\
 &  & V_{j+1}^{\top}\phi_{j+1}^{\prime}\left(\Phi_{j}\right)\left(I_{L_{j+1}}\otimes\varphi_{j}^{\top}\right)\\
 &  & +V_{w}^{\top}\phi_{w}^{\prime}\left(\Phi_{w-1}\right)\ldots\left(I_{L_{q+1}}\otimes\phi_{q}^{\prime\top}\left(\Phi_{q-1}\right)\right)\\
 &  & \ldots V_{j+1}^{\top}\phi_{j+1}^{\prime}\left(\Phi_{j}\right)\left(I_{L_{j+1}}\otimes\varphi_{j}^{\top}\right).
\end{eqnarray*}
Taking norms and applying the triangle inequality yields
\begin{eqnarray}
\left\Vert \frac{\partial^{2}\Phi_{w}}{\partial\mathrm{vec}\left(V_{q}\right)\partial\mathrm{vec}\left(V_{j}\right)}\right\Vert  & \leq & \left\Vert V_{w}\right\Vert \left\Vert \phi_{w}^{\prime\prime}\left(\Phi_{w-1}\right)\right\Vert \left\Vert \frac{\partial\Phi_{w-1}}{\partial\mathrm{vec}\left(V_{q}\right)}\right\Vert \nonumber \\
 &  & \left\Vert V_{w-1}\right\Vert \left\Vert \phi_{w-1}^{\prime}\left(\Phi_{w-2}\right)\right\Vert \ldots\nonumber \\
 &  & \left\Vert V_{j+1}\right\Vert \left\Vert \phi_{j+1}^{\prime}\left(\Phi_{j}\right)\right\Vert \left\Vert \varphi_{j}\right\Vert \nonumber \\
 &  & +\left\Vert V_{w}\right\Vert \left\Vert \phi_{w}^{\prime}\left(\Phi_{w-1}\right)\right\Vert \left\Vert V_{w-1}\right\Vert \nonumber \\
 &  & \left\Vert \phi_{w-1}^{\prime\prime}\left(\Phi_{w-2}\right)\right\Vert \left\Vert \frac{\partial\Phi_{w-2}}{\partial\mathrm{vec}\left(V_{q}\right)}\right\Vert \nonumber \\
 &  & \left\Vert V_{w-2}\right\Vert \left\Vert \phi_{w-2}^{\prime}\left(\Phi_{w-3}\right)\right\Vert \ldots\nonumber \\
 &  & \left\Vert V_{j+1}\right\Vert \left\Vert \phi_{j+1}^{\prime}\left(\Phi_{j}\right)\right\Vert \left\Vert \varphi_{j}\right\Vert \nonumber \\
 &  & \vdots\nonumber \\
 &  & +\left\Vert V_{w}\right\Vert \left\Vert \phi_{w}^{\prime}\left(\Phi_{w-1}\right)\right\Vert \ldots\left\Vert V_{q+1}\right\Vert \nonumber \\
 &  & \left\Vert \phi_{q+1}^{\prime\prime}\left(\Phi_{q}\right)\right\Vert \left\Vert \frac{\partial\Phi_{q}}{\partial\mathrm{vec}\left(V_{q}\right)}\right\Vert \nonumber \\
 &  & \left\Vert V_{q}\right\Vert \left\Vert \phi_{q}^{\prime}\left(\Phi_{q-1}\right)\right\Vert \ldots\nonumber \\
 &  & \left\Vert V_{j+1}\right\Vert \left\Vert \phi_{j+1}^{\prime}\left(\Phi_{j}\right)\right\Vert \left\Vert \varphi_{j}\right\Vert \nonumber \\
 &  & +\left\Vert V_{w}\right\Vert \left\Vert \phi_{w}^{\prime}\left(\Phi_{w-1}\right)\right\Vert \ldots\left\Vert \phi_{q}^{\prime}\left(\Phi_{q-1}\right)\right\Vert \nonumber \\
 &  & \ldots\left\Vert V_{j+1}\right\Vert \left\Vert \phi_{j+1}^{\prime}\left(\Phi_{j}\right)\right\Vert \left\Vert \varphi_{j}\right\Vert .\label{eq:partial Hessian bounds 2}
\end{eqnarray}
We express the bounds from Lemma \ref{lem:Jacobian bound} using
the definition of $\mathcal{Q}_{j}$ as
\begin{align*}
\left\Vert \frac{\partial\Phi_{w}}{\partial\mathrm{vec}\left(V_{j}\right)}\right\Vert  & \leq\mathfrak{b}_{0}^{w-j}\left(\prod_{l=j+1}^{w}\left\Vert V_{l}\right\Vert \right)\mathcal{Q}_{j},\\
\left\Vert \frac{\partial\Phi_{w}}{\partial\mathrm{vec}\left(V_{q}\right)}\right\Vert  & \leq\mathfrak{b}_{0}^{w-q}\left(\prod_{l=q+1}^{w}\left\Vert V_{l}\right\Vert \right)\mathcal{Q}_{q}.
\end{align*}
Applying these bounds to (\ref{eq:partial Hessian bounds 2}) yields
\begin{eqnarray*}
\left\Vert \frac{\partial^{2}\Phi_{w}}{\partial\mathrm{vec}\left(V_{q}\right)\partial\mathrm{vec}\left(V_{j}\right)}\right\Vert  & \leq & \mathfrak{c}_{0}\mathfrak{b}_{0}^{w-j-1}\left(\prod_{l=j+1}^{w}\left\Vert V_{l}\right\Vert \right)\mathcal{Q}_{j}\\
 &  & \qquad\left(\mathfrak{b}_{0}^{w-q-1}\left(\prod_{l=q+1}^{w-1}\left\Vert V_{l}\right\Vert \right)\mathcal{Q}_{q}\right)\\
 &  & +\mathfrak{c}_{0}\mathfrak{b}_{0}^{w-j-1}\left(\prod_{l=j+1}^{w}\left\Vert V_{l}\right\Vert \right)\mathcal{Q}_{j}\\
 &  & \qquad\left(\mathfrak{b}_{0}^{w-q-2}\left(\prod_{l=q+1}^{w-2}\left\Vert V_{l}\right\Vert \right)\mathcal{Q}_{q}\right)\\
 &  & \vdots\\
 &  & +\mathfrak{c}_{0}\mathfrak{b}_{0}^{w-j-1}\left(\prod_{l=j+1}^{w}\left\Vert V_{l}\right\Vert \right)\mathcal{Q}_{j}\\
 &  & \qquad\left(\mathfrak{b}_{0}^{0}\left(\prod_{l=q+1}^{q}\left\Vert V_{l}\right\Vert \right)\mathcal{Q}_{q}\right)\\
 &  & +\mathfrak{b}_{0}^{w-j+1}\left(\prod_{l=j+1}^{w}\left\Vert V_{l}\right\Vert \right)\mathcal{Q}_{j}.
\end{eqnarray*}
Factoring the common term $\left(\prod_{l=j+1}^{w}\left\Vert V_{l}\right\Vert \right)\mathcal{Q}_{j}$
out of every line and combining the $\mathfrak{b}_{0}$ exponents
$\left(w-j-1\right)+\left(w-q-h\right)=2w-q-j-h-1$ in the $h$-th line of the sum yields
\begin{eqnarray*}
\left\Vert \frac{\partial^{2}\Phi_{w}}{\partial\mathrm{vec}\left(V_{q}\right)\partial\mathrm{vec}\left(V_{j}\right)}\right\Vert  & \leq & \left(\prod_{l=j+1}^{w}\left\Vert V_{l}\right\Vert \right)\mathcal{Q}_{j}\bigg[\sum_{h=1}^{w-q}\mathfrak{c}_{0}\mathfrak{b}_{0}^{2w-q-j-h-1}\\
 &  & \quad\bigg(\prod_{l=q+1}^{w-h}\left\Vert V_{l}\right\Vert \bigg)\mathcal{Q}_{q}+\mathfrak{b}_{0}^{w-j+1}\bigg].
\end{eqnarray*}
Identifying the bracketed expression with the definitions of $\mathcal{R}_{w,q,j}$
and $\mathcal{T}_{w,j}$ gives
\[
\text{\ensuremath{\left\Vert \frac{\partial^{2}\Phi_{w}}{\partial\mathrm{vec}\left(V_{q}\right)\partial\mathrm{vec}\left(V_{j}\right)}\right\Vert }}\leq\mathcal{R}_{w,q,j}\mathcal{Q}_{j}\mathcal{Q}_{q}+\mathcal{T}_{w,j}\mathcal{Q}_{j}.
\]
Because the terms $\mathcal{Q}_{j}$ and $\mathcal{Q}_{q}$ are linear
in $\left\Vert \sigma_{a}\right\Vert $, it follows that the bound
on $\left\Vert \frac{\partial^{2}\Phi_{w}}{\partial\mathrm{vec}\left(V_{q}\right)\partial\mathrm{vec}\left(V_{j}\right)}\right\Vert $
is quadratic in $\left\Vert \sigma_{a}\right\Vert $.
\end{proof}

\section{Bounds on Higher-Order Terms in First-Order Taylor Series Approximation}

The bounds obtained in the previous section can be used to compute
analytical bounds on the first-order Taylor series approximation of
$\text{\ensuremath{\theta}}\mapsto\Phi\left(\sigma,\theta\right)$.
To this end, consider a compact convex set of admissible DNN parameters
given by $\Theta\subset\mathbb{R}^{p}$. Additionally, consider any
arbitrary elements $\theta^{*},\hat{\theta}\in\Theta$, and let $\tilde{\theta}\triangleq\theta^{*}-\hat{\theta}\in\mathbb{R}^{p}$
denote their difference. Furthermore, let $\Phi^{(i)}$ denote the
$i^{th}$ element of $\Phi$ for all $i\in\{1,\ldots,n\}$, where
$n=L_{\textrm{out}}$. Then performing
a first-order application of Taylor's theorem \cite[Theorem 4.7]{Lax2017}
element-wise to $\Phi$ yields
\begin{equation}
\Phi(\sigma,\theta^{*})-\Phi\left(\sigma,\hat{\theta}\right)=\frac{\partial\Phi\left(\sigma,\hat{\theta}\right)}{\partial\hat{\theta}}\tilde{\theta}+R\left(\sigma,\tilde{\theta}\right),\label{eq:Taylor_approx}
\end{equation}
% CORRECTION: domain was R^{2n} x R^p; corrected to R^{L_in} x R^p.
where $R:\mathbb{R}^{L_{\textrm{in}}}\times\mathbb{R}^{p}\to\mathbb{R}^{n}$ denotes
the Lagrange remainder term given by
\[
\begin{aligned}
R\left(\sigma,\tilde{\theta}\right) & =\frac{1}{2}\bigg[\tilde{\theta}^{\top}\frac{\partial^{2}\Phi^{(1)}}{\partial\hat{\theta}^{2}}\left(\sigma,\theta^{*}+\varpi_{1}\left(\hat{\theta}\right)\tilde{\theta}\right)\tilde{\theta},\\
 & \qquad\ldots,\tilde{\theta}^{\top}\frac{\partial^{2}\Phi^{(n)}}{\partial\hat{\theta}^{2}}\left(\sigma,\theta^{*}+\varpi_{n}\left(\hat{\theta}\right)\tilde{\theta}\right)\tilde{\theta}\bigg]^{\top},
\end{aligned}
\]
where $\varpi_{i}:\mathbb{R}^{p}\to\left[0,1\right]$ denote unknown
functions parameterizing a convex combination of $\theta^{*}$ and
$\hat{\theta}$ for all $i\in\{1,\ldots,n\}$. The following theorem
provides a polynomial bound on the Lagrange remainder term.
\begin{thm}
\label{thm:Lagrange remainder bound} There exists a polynomial function
$\rho_{0}:\mathbb{R}_{\geq0}\to\mathbb{R}_{\geq0}$ of the form $\rho_{0}\left(\left\Vert \sigma\right\Vert \right)=a_{2}\left\Vert \sigma\right\Vert ^{2}+a_{1}\left\Vert \sigma\right\Vert ^{1}+a_{0}$
with constants $a_{2},a_{1},a_{0}\in\mathbb{R}_{>0}$ such that the
Lagrange remainder term is bounded as $\left\Vert R\left(\sigma,\tilde{\theta}\right)\right\Vert \leq\rho_{0}\left(\left\Vert \sigma\right\Vert \right)\left\Vert \tilde{\theta}\right\Vert ^{2}$.
\end{thm}
\begin{proof}
Using the triangle inequality, it follows that
\begin{equation}
\begin{aligned}
\left\Vert R\left(\sigma,\tilde{\theta}\right)\right\Vert  & \leq\frac{\left\Vert \tilde{\theta}\right\Vert ^{2}}{2}\bigg(\left\Vert \frac{\partial^{2}\Phi^{(1)}}{\partial\hat{\theta}^{2}}\left(\sigma,\theta^{*}+\varpi_{1}\left(\hat{\theta}\right)\tilde{\theta}\right)\right\Vert \\
 & \qquad+\ldots+\left\Vert \frac{\partial^{2}\Phi^{(n)}}{\partial\hat{\theta}^{2}}\left(\sigma,\theta^{*}+\varpi_{n}\left(\hat{\theta}\right)\tilde{\theta}\right)\right\Vert \bigg).
\end{aligned}
\label{eq:Elementwise Hessian Bounds}
\end{equation}
Due to the facts that $\Theta$ is a convex set, $\varpi_{i}\left(\hat{\theta}\right)\in\left[0,1\right]$,
and $\theta^{*},\hat{\theta}\in\Theta$, it follows that $\theta^{*}+\varpi_{i}\left(\hat{\theta}\right)\tilde{\theta}\in\Theta$
for all $i\in\{1,\ldots,n\}$ and $\hat{\theta}\in\Theta$. Because
$\varpi_{i}$ is unknown, the term $\left\Vert \frac{\partial^{2}\Phi^{(i)}}{\partial\hat{\theta}^{2}}\left(\sigma,\theta^{*}+\varpi_{i}\left(\hat{\theta}\right)\tilde{\theta}\right)\right\Vert $
is bounded in the following analysis considering the worst case bound
in which $\left\Vert \theta^{*}+\varpi_{i}\left(\hat{\theta}\right)\tilde{\theta}\right\Vert \leq\bar{\theta}$
for all $i\in\{1,\ldots,n\}$. To this end, a bound is developed on
$\frac{\partial^{2}\Phi^{(i)}}{\partial\hat{\theta}^{2}}\left(\sigma,\theta\right)$
considering an arbitrary $\left\Vert \theta\right\Vert \leq\bar{\theta}$.
The Hessian matrix $\frac{\partial^{2}\Phi^{(i)}}{\partial\theta^{2}}\left(\sigma,\theta\right)$
is comprised of blocks of the form $\frac{\partial^{2}\Phi^{(i)}\left(\sigma,\theta\right)}{\partial\mathrm{vec}(V_{q})\partial\mathrm{vec}(V_{j})}$
for all $q,j\in\left\{ 0,\ldots,k\right\} $ and $i\in\left\{ 1,\dots,n\right\} $.
Writing the Hessian as the sum of its $\left(q,j\right)$-blocks (each
embedded at the appropriate position with zeros elsewhere) and applying
the triangle inequality bounds the spectral norm by the sum of block
norms. Each block satisfies $\left\Vert \frac{\partial^{2}\Phi^{(i)}}{\partial\mathrm{vec}\left(V_{q}\right)\partial\mathrm{vec}\left(V_{j}\right)}\right\Vert \leq\mathcal{R}_{k,q,j}\mathcal{Q}_{j}\mathcal{Q}_{q}+\mathcal{T}_{k,j}\mathcal{Q}_{j}$
by Lemma \ref{lem:Hessian bounds} with $w=k$ (since $\Phi=\Phi_{k}$).
Therefore, using (\ref{eq:Elementwise Hessian Bounds}), the Lagrange
remainder can be bounded as
\begin{equation}
\left\Vert R\left(\sigma,\tilde{\theta}\right)\right\Vert \leq\frac{\left\Vert \tilde{\theta}\right\Vert ^{2}}{2}\left(\sum_{i=1}^{n}\sum_{q=0}^{k}\sum_{j=0}^{k}\left(\mathcal{R}_{k,q,j}\mathcal{Q}_{j}\mathcal{Q}_{q}+\mathcal{T}_{k,j}\mathcal{Q}_{j}\right)\right),\label{eq:R bound 2}
\end{equation}
where the terms $\mathcal{Q}_{j},\mathcal{Q}_{q},\mathcal{T}_{k,j},$
and $\mathcal{R}_{k,q,j}$ are defined in Lemma \ref{lem:Hessian bounds}.
As a result, defining the polynomial function $\rho_{0}$ as
\begin{equation}
\rho_{0}\left(\left\Vert \sigma\right\Vert \right)\triangleq\frac{1}{2}\sum_{i=1}^{n}\sum_{q=0}^{k}\sum_{j=0}^{k}\left(\mathcal{R}_{k,q,j}\mathcal{Q}_{j}\mathcal{Q}_{q}+\mathcal{T}_{k,j}\mathcal{Q}_{j}\right)\label{eq:rho 0 function explicit}
\end{equation}
yields $\left\Vert R\left(\sigma,\tilde{\theta}\right)\right\Vert \leq\rho_{0}\left(\left\Vert \sigma\right\Vert \right)\left\Vert \tilde{\theta}\right\Vert ^{2}$.
Because Lemma \ref{lem:Hessian bounds} provides quadratic polynomial
bounds on such terms with respect to the input, it follows that $\rho_{0}$
is a quadratic function of the form $\rho_{0}\left(\left\Vert \sigma\right\Vert \right)=a_{2}\left\Vert \sigma\right\Vert ^{2}+a_{1}\left\Vert \sigma\right\Vert ^{1}+a_{0}$
with constants $a_{2},a_{1},a_{0}\in\mathbb{R}_{>0}$. We note that
the explicit values of $a_{2},a_{1},a_{0}$ can be computed by expanding
the polynomial formulation in (\ref{eq:rho 0 function explicit}).
\end{proof}

\bibliographystyle{ieeetr}
\bibliography{encr,ncr,master}

\end{document}